\documentclass[11pt]{article}
\usepackage{setspace}
\usepackage[T1]{fontenc}
\usepackage[utf8]{inputenc}
\usepackage{authblk}
\usepackage{fullpage}

\usepackage{float}
\usepackage{tabularx}
\usepackage{amsthm}
\usepackage{amsmath}
\usepackage{amssymb}
\usepackage{amsfonts}
\usepackage{mathtools}
\usepackage{enumitem}
\usepackage{comment}
\usepackage[pagebackref]{hyperref}
\usepackage{url}
\usepackage[sort,numbers]{natbib}

\usepackage{multirow}
\usepackage{tikz}
\usetikzlibrary{mindmap}
\usetikzlibrary{arrows}

\usepackage[capitalize]{cleveref}

\usepackage{todonotes}

\usepackage{thm-restate}

\usepackage[font=small]{caption} 
\newtheorem{theorem}{Theorem}
\newtheorem{lemma}{Lemma}

\newtheorem{corollary}{Corollary}

\theoremstyle{definition}

\usepackage{thmtools}
\usepackage{thm-restate}

\newcommand{\threefield}[3]{$#1\mid#2\mid#3$}

\newcounter{dvircounter}
\newcounter{dannycounter}
\newcommand{\Dvir}[1]{\noindent\textcolor{red}{$\ll$\textsf{#1}$\gg$\marginpar{\tiny\bf 
\textcolor{red}{\refstepcounter{dvircounter}Dvir~\thedvircounter}}}}
\newcommand{\Danny}[1]{\noindent\textcolor{blue}{$\ll$\textsf{#1}$\gg$\marginpar{\tiny\bf 
\textcolor{blue}{\refstepcounter{dannycounter}Danny~\thedannycounter}}}}

\title{Fast Makespan Minimization via Short ILPs}

\author[1]{Danny~Hermelin}
\author[1]{Dvir~Shabtay}

\affil[1]{\small Department of Industrial Engineering and Management, Ben-Gurion University of the Negev, Israel,
\texttt{hermelin@bgu.ac.il, dvirs@bgu.ac.il}}

\date{}

\begin{document}

\maketitle

\abstract{
Short integer linear programs are programs with a relatively small number of constraints. We show how recent improvements on the running-times of solvers for such programs can be used to obtain fast pseudo-polynomial time algorithms for makespan minimization on a fixed number of parallel machines, and other related variants. The running times of our algorithms are all of the form $\widetilde{O}(p^{O(1)}_{\max}+n)$ or $\widetilde{O}(p^{O(1)}_{\max} \cdot n)$, where $p_{\max}$ is the maximum processing time in the input. These improve upon the time complexity of previously known algorithms for moderate values of $p_{\max}$. 
}


\section{Introduction}
\label{sec:introduction}%

Recent breakthroughs on the running times of algorithms for solving integer linear programs (ILPs) have made a significant impact on algorithm design. So far, there have been essentially two strands of this research enterprise. The first began with Lenstra~\cite{Lenstra83}, who showed that ILPs with a parameterized number of variables, aka \emph{thin ILPs}, can be solved in fixed-parameter time. Over the years, there have been several improvements to this result~\cite{DadushPV11,HildebrandKoppe2013,Kannan87}, culminating in an algorithm with time $\log^{O(N)} N \cdot \log^{O(1)}\Delta$ for ILPs with $N$ variables and encoding size~$\Delta$~\cite{ReisRothvoss23}. The second strand began with Papadimitriou~\cite{Papadimitriou81} who showed that ILPs with a fixed number $m$ of constraints, or \emph{short ILPs}, can be solved in pseudopolynomial-time. In a breakthrough result, Eisenbrand and Weismantel~\cite{EisenbrandWeismantel2020} improved this using the Steinitz Exchange Lemma to fixed-parameter time when parameterizing by the number of constraints and the maximum lefthand-side coefficient $a_{\max}$. Recently, Jansen and Rohwedder~\cite{JansenRohwedder} improved this result for unbounded ILPs (see Section~\ref{sec:ShortILPs}), culminating in an elegant $O((\sqrt{M} \cdot a_{\max})^{2M} + MN)$-time algorithm for ILPs with $M$ constraints and maximal coefficient constraint~$a_{\max}$. 

In this paper, we demonstrate how short ILP solvers can be leveraged to obtain efficient pseudo-polynomial time algorithms for makespan minimization and related scheduling problems. While previous research in this domain has primarily focused on solvers for thin ILPs (\emph{e.g.}~\cite{HeegerM25,Hermelin2019,Hermelin2022,Karhi21,KaulMM24,Mnich2015}), Jansen and Rohwedder showed that their solver can yield a fast EPTAS for minimizing the makespan on a set identical parallel machines~\cite{JansenRohwedder}. We extend this line of research by showing that short ILP solvers are equally powerful for constructing fast pseudo-polynomial time algorithms for both the basic makespan problem and its generalizations. 

\paragraph{The standard ILP for makespan minimization.}

Recall that in the \emph{makespan minimization} problem, we are given a set $J=[n]$\footnote{We use $[n]$ to denote the set of natural numbers $\{1,\ldots,n\}$ for any $n \in \mathbb{N}$} of~$n$ jobs to be processed non-preemptively on a set of $m$ parallel machines $M_1,\ldots,M_m$, with the goal of minimizing the completion time of last job to be processed (\emph{i.e.} the makespan). Throughout the paper, we will focus on the case where $m=O(1)$; that is, the number of machines is a fixed constant independent of the number of jobs~$n$. The processing time of job~$j$ on machine $M_i$ is denoted by $p_{i,j} \in \mathbb{N}$. A \emph{schedule} is an assignment $\sigma: [n] \to [m]$ of jobs to machines, and the \emph{makespan} of machine $M_i$ under $\sigma$ is defined as $M_i(\sigma)=\sum_{j\in\sigma^{-1}(i)}p_{i,j}$. Our goal is to compute a schedule with minimum makespan, where the makespan of the schedule $\sigma$ is $C_{\max}(\sigma)=\max_{i \in [m]}M_i(\sigma)$. This problem is denoted by \threefield{Rm}{}{C_{\max}} in the classical three-field notation of scheduling problems~\cite{GrahamEtAl79}. The special case of this problem where $p_{i,j}=p_j/q_i$, and $q_i$ is the \emph{machine speed} of~$M_i$, is called the \emph{makespan minimization on related machines} problem, and is denoted by \threefield{Qm}{}{C_{\max}}.

The standard textbook ILP formulation for \threefield{Rm}{}{C_{\max}} has $nm$ binary variables $x_{i,j} \in \{0,1\}$, one for each for job-machine pair, where $x_{i,j} =1$ indicates that job~$j$ is scheduled on machine~$M_i$ (\emph{i.e.}, $\sigma(j)=i$). Letting $C$ denote the target makespan, the constraints of this ILP are then given by:
\begin{alignat}{3}
\tag{$\text{ILP}_0$}
& &\quad & \sum_{j \in [n]} p_{i,j}x_{i,j}  \;\leq\; C & \qquad \qquad & \forall \, i \in [m] \notag  \\
& &\quad & \sum_{i \in [m]} x_{i,j}  \;=\; 1 & \qquad \qquad & \forall \, j \in [n] \notag  \\
& \qquad \qquad & & x_{i,j} \in  \{0,1\} & \qquad \qquad & \forall \, i \in [m],\, j \in[n]\notag
\end{alignat}
Any feasible solution to this ILP translates naturally to a schedule $\sigma$ of makespan at most $C$ by setting $\sigma(j)=i$ iff $x_{i,j}=1$ for all jobs $j\in J$. This formulation was used in the seminal paper by Lenstra, Shmoys, and Tardos~\cite{LenstraShmoysTardos90} to obtain an FPTAS for \threefield{Rm}{}{C_{\max}}, along with a 2-approximation for the case where $m$ is unbounded.

\paragraph{Our contribution.} 

Observe that the constraint matrix of the standard ILP for \threefield{Rm}{}{C_{\max}} has $n+m$ rows, which means that it is not short even though $m=O(1)$. We show that when considering the special case \threefield{Qm}{}{C_{\max}} problem, one can replace the $n$ constraints of the form $\sum_i x_{i,j} =1$ by a single constraint $\sum_i \sum_j x_{i,j} =n$ along with a carefully constructed objective function. This gives a constraint matrix with $m+1$ rows and maximum coefficient $p_{\max}$, and so we can use the short ILP solver of Jansen and Rohwedder~\cite{JansenRohwedder} to obtain the main result of the paper:

\begin{restatable}{theorem}{Cmax}
\label{thm:main}%
$Qm||C_{\max}$ can be solved in $\widetilde{O}(p^{2m+2}_{\max} + n)$ time.
\end{restatable}

We then turn to consider three generalizations of \threefield{Qm}{}{C_{\max}}. In the first, denoted as the \threefield{Qm}{cap}{C_{\max}} problem~\cite{Woeginger05,YangYeZhang2003}, each machine has a \emph{capacity}~$b_i$, and a feasible schedule $\sigma$ is required to have $|\sigma^{-1}(i)| \leq b_i$ for all $i\in[m]$. In the \threefield{Qm}{rej}{C_{\max}} problem, each job~$j$ has a \emph{weight}~$w_j \in \mathbb{N}$ and we are allowed to reject jobs of total weight at most some prespecified $W \in \mathbb{N}$. That is, we search for a subset of jobs $J' \subseteq J$ that admits a schedule of makespan at most $C$ where $\sum_{j \in J \setminus J'} w_j \leq W$. In the third variant, the \threefield{Qm}{r_j \geq 0}{C_{\max}} problem, each job~$j$ has a \emph{release date}~$r_j \in \mathbb{N}$, and jobs are required to be scheduled no earlier than the release dates (this requires a slightly different definition of a schedule, see Section~\ref{sec:ReleaseDates}). We show that the ILP we use for \threefield{Qm}{}{C_{\max}} can be adapted to each of these variants. 

For the case of unrelated machines, we consider the special case where $p_{i,j} \in \{p_j, \infty\}$ for all $i \in [m]$ and $j \in [n]$; that is, the processing time of each job $j$ is $p_j$ on all machines, but some machines cannot process $j$. For this variant, denoted by \threefield{Rm}{p_{i,j}\in\{p_j,\infty\}}{C_{\max}}, we provide an ILP with an exponential number of constraints in~$m$, which is still constant since $m=O(1)$. Finally, for the special case of two machines, the \threefield{R2}{}{C_{\max}} problem, we show that an alternative ILP formulation yields an $\widetilde{O}(p^{6}_{\max} \cdot n)$ time algorithm. Table~\ref{tab:results} summarizes our results, and compares them with the previously best bounds that are discussed below.

\begin{table}[bt]
\begin{center}
\begin{tabular}{c|c|c}
 \quad  Variant  \quad  & \quad Previous Fastest \quad & New Result  \quad   \\
\hline
\hline
\threefield{Qm}{}{C_{\max}} & $\widetilde{O}(p^{m-1}_{\max} \cdot n^{m-1})$ & $\widetilde{O}(p^{2m+2}_{\max} + n)$ \\ 
\hline
\threefield{Qm}{cap}{C_{\max}} & $O(p^{m-1}_{\max}\cdot n^{2m-2})$~\cite{Woeginger05} & $\widetilde{O}(p^{4m+2}_{\max} + n)$ \\ 
\hline
\threefield{Qm}{rej}{C_{\max}} & $O(p^m_{\max} \cdot n^m)$~\cite{Zhang09} & $\widetilde{O}((p_{\max}+w_{\max})^{(m+2)(m+3)} \cdot n)$ \\ 
\hline
\threefield{Qm}{r_j \geq 0}{C_{\max}} & $O(p^{m-1}_{\max} \cdot n^{m-1})$~\cite{LawlerLenstraKanShmoys1993} & $\widetilde{O}(p^{2r_{\#}(m+1)}_{\max} + r_{\#}n)$ \\ 
\hline
\hline
\threefield{Rm}{p_{i,j} \in \{p_j,\infty\}}{C_{\max}} & $O(p^{m-1}_{\max} \cdot n^{m-1})$~\cite{Sahni76} & $\widetilde{O}(p^{2^{m+1}+2m-2)}_{\max} + n)$ \\ 
\hline
\threefield{R2}{}{C_{\max}} & $O(p_{\max} \cdot n^2)$~\cite{HorowitzSahni76} & $\widetilde{O}(p^{6}_{\max} \cdot n)$ \\ 
\end{tabular}
\end{center}
\caption[]{A summery of our results compared with the previously fastest known algorithms, in the setting where only $p_{\max}=\max_{i,j} p_{i,j}$ and $n$ are taken as parameters. We use~$r_{\#}$ to denote the number of distinct release dates $|\{r_1,\ldots,r_n\}|$ of the input set of jobs. The first result of the second column is due to fast polynomial multiplication.}
\label{tab:results}
\end{table}

\paragraph{Previous Work.} 

The problems studied in this paper can be viewed as extensions of \threefield{Q2}{}{C_{\max }}, which is known to be NP-hard even when machines have equal unit speeds (\emph{i.e.}, the \threefield{P2}{}{C_{\max }} problem)~\cite{LenstraRinnooy77}. This hardness stems from its fundamental equivalence to the partition and subset sum problems. Historically, two primary techniques have been used to obtain pseudopolynomial-time algorithms for \threefield{Qm}{}{C_{\max }} and \threefield{Rm}{}{C_{\max }}. The first is dynamic programming, pioneered by Sahni~\cite{Sahni76} and Horowitz and Sahni~\cite{HorowitzSahni76}, which yields a running time of $O(P^{m-1}\cdot n)$, where $P$ is the maximum total processing time on any machine. The second leverages fast polynomial multiplication (FFT) to achieve a running time of $\widetilde{O}(P^{m-1})$\footnote{The notation $\widetilde{O}()$ is used to suppress polylogarithmic factors.}. This approach was applied to subset sum by Koiliaris and Xu~\cite{KoiliarisXu2019}, and their ideas generalize naturally to the \threefield{Qm}{}{C_{\max}} problem, and to the \threefield{Rm}{}{C_{\max}} problem at the cost of an additional multiplicative factor of~$P$ to the time complexity. 

While polynomial multiplication does not appear to extend to the other variants considered here, Sahni’s dynamic programming approach handles them more naturally. For \threefield{Qm}{cap}{C_{\max}}, this framework yields an $O(P^{m-1}\cdot n^{m-1})$ time algorithm, as noted by Woeginger~\cite{Woeginger05}. For \threefield{Qm}{rej}{C_{\max}}, Zhang \emph{et al.}~\cite{Zhang09} provide an $O(P^{m}\cdot n)$ time algorithm, while the \threefield{Qm}{r_j \geq}{C_{\max}} variant can be solved in $O(P^{m-1}\cdot n)$ time as shown in~\cite{LawlerLenstraKanShmoys1993}. Together, all results discussed above provide the baseline for evaluating our new algorithms.

Our algorithms outperform the above results in cases where $p_{\max}$ is sublinear in~$n$. For example, when $p_{\max}=O(\sqrt{n})$, the dynamic programming of Sahni~\cite{Sahni76} runs in $O(n^{\frac{3m-1}{2}})$ time, whereas our algorithm achieves a superior $\widetilde{O}(n^{m+1})$ running time for $m > 4$. The runtime obtained by fast polynomial multiplication is $O(n^{\frac{3}{2}(m-1)})$, which is slower than our algorithm for $m>5$. When $p_{\max}=O(n^{1/4})$, our algorithm becomes faster already for $m>2$. Generally speaking, in the setting where only $p_{\max}$ and $n$ are taken as parameters, we can translate the upper bounds above by replacing~$P$ with $p_{\max} \cdot n$, as is done in Table~\ref{tab:results}. As a final note, we mention that at the cost of an increase to their running times, most of our algorithms can be easily adapted so that they only use  polynomial space, using the recent space-efficient short ILP solver of Rohwedder and Wegrzycki~\cite{RohwedderWegrzycki25}. In contrast, both the dynamic programming approach, as well as the polynomial multiplication approach, seem to inherently require exponential space.

\section{Short ILPs}
\label{sec:ShortILPs}%

We next briefly review the state of the art regarding short ILPs, and present the ILP solvers that will be used throughout the paper. Let $A \in \mathbb{Q}^{M \times N}$ be an $M \times N$ matrix of rational values,~$b \in \mathbb{Q}^{M}$ be a rational vector of dimension $M$, and $c \in \mathbb{Q}^{N}$ be a rational vector of dimension $N$. Also, let $u \in \mathbb{N}^{N}$ be a vector of non-negative integers of dimension $N$. The \textsc{ILP-bounded} problem asks to determine the minimum value of $c^\top x$ over all natural vectors~$x$ of dimension~$N$ that satisfy $Ax=b$ and $x \leq u$. That is, 
$$
\textsc{ILP-bounded} = \min \{c^\top x: x \in \mathbb{N}^N,\, x \leq u,\text{  and  }Ax=b\}.
$$
Here and throughout we assume that $0 \in\mathbb{N}$. Let $a_{\max}$ denote the maximal absolute value of any entry in $A$. The seminal paper of Eisenbrand and Weismantel proves the following:
\begin{theorem}[\cite{EisenbrandWeismantel2020}]
\label{thm:ILP-bounded}%
\textsc{ILP-bounded} can be solved in $\widetilde{O}(M^{(M+1)^2} \cdot a_{\max}^{M(M+1)} \cdot n)$ time.  
\end{theorem}
\noindent In~\cite{KnopPW20}, it is shown that even for binary constraint matrices $A$, there is no $2^{o(M \log M)} \cdot (n + ||b||_{\infty})$ time algorithm for \textsc{ILP-bounded}, assuming the Exponential Time Hypothesis.  

Jansen and Rohwedder~\cite{JansenRohwedder} showed that the algorithm of Eisenbrand and Weismantel can be substantially improved if one removes the upper-bound constraint vector $u \in \mathbb{N}^{N}$ from the above ILP formulation. More specifically, let the \textsc{ILP-unbounded} problem be defined as follows:
$$
\textsc{ILP-unbounded} = \min \{c^\top x: x \in \mathbb{N}^N \text{  and  }Ax=b\}.
$$
Jansen and Rohwedder proved the following: 
\begin{theorem}[\cite{JansenRohwedder}]
\label{thm:ILP-unbounded}%
\textsc{ILP-unbounded} can be solved in $\widetilde{O}(M^M \cdot a_{\max}^{2M}+ MN)$ time. 
\end{theorem}


\section{Makespan minimization on related machines}
\label{sec:Makespan}

In the following section we show that \threefield{Qm}{}{C_{\max}} on $n$ jobs with maximum processing time~$p_{\max}$ can be solved in $\widetilde{O}(p^{2(m+1)}_{\max}+n)$ time, providing a proof for our main result Theorem~\ref{thm:main}. Throughout the section, we assume that the jobs are ordered according to the Shortest Processing Time (SPT) rule, and so $p_1\leq p_{2}\leq \cdots \leq p_n$. Our algorithm determines the optimal makespan value~$C$ by performing a binary search within the range $[n,\,p_{\max }\cdot n]$, where $n$ is a lower bound and $np_{\max}$ is an upper bound on the minimal makespan value. The $O(n\log n)$ additive cost for sorting and the $O(\log p_{\max} + \log n)$ multiplicative factor for the search iterations represent a negligible overhead to the overall running time. Thus, we can proceed as if the value~$C$ is given, simplifying the problem to its decision version. 

Recall the ILP formulation for \threefield{Rm}{}{C_{\max}} described in Section~\ref{sec:introduction}. We will refer to this program as $\text{ILP}_0$ throughout. Note that $\text{ILP}_0$ is not short, since the number of rows in the constraint matrix depends also on $n$, due to the $n$ constraints of the form $\sum_i x_{i,j}=1$. We construct a short ILP which we refer to as $\text{ILP}_1$, such that the feasible solutions of $\text{ILP}_1$ with non-negative objective values correspond to feasible solutions of $\text{ILP}_0$. The program $\text{ILP}_1$ again has $mn$ integer variables $x_{1,1},\ldots,x_{m,n} \in \mathbb{N}$, one for each job-machine pair, but here these variables can take any integer value, and not only values in $\{0,1\}$. Nevertheless, our intended interpretation for these variables remains the same; that is, $x_{i,j}=1$ corresponds to the case where job~$j$ is assigned to $M_i$. We also add $m$ integer slack variables $s_1,\ldots,s_{m} \in \mathbb{N}$, one per machine, where variable~$s_i$ will correspond to the difference between the makespan value~$C$ and the total load on machine $M_i$. 

Next we turn to describing the constraints in $\text{ILP}_1$. First, we convert all inequality constraints in $\text{ILP}_0$ to equality constraints in standard fashion. That is, for each $i \in [m]$, we construct a constraint corresponding to machine $M_i$ by:
$$
\sum_{j\in [n]}  p_j  x_{i,j} \,+\, s_{i} \;=\; q_iC.
$$
Next, we replace all the $n$ equality constraints in $\text{ILP}_0$ by a single constraint by summing up all these constraints. That is, we add the following constraint to $\text{ILP}_1$, which ensures that the total sum of the variables in $\text{ILP}_1$ is exactly~$n$:
$$
\sum_{i\in [m]} \sum_{j\in [n]}   x_{i,j}  \;=\; n.
$$
Thus, in total, $\text{ILP}_1$ has $m+1$ constraints.

Note the constraint $\sum_i \sum _j$ still allows assigning jobs multiple times, \emph{i.e.} we could still have $\sum_i x_{i,j}  > 1$ and $\sum_i x_{i,k}  = 0$ for some $j \neq k \in [n]$. One way of fixing this is to add the following constraint 
$$
\sum_{j\in[n]}  n^{j}\sum_{i\in [m]}  x_{i,j} \,=\, \sum_{j\in [n]}  n^{j}.
$$
Indeed, the only way to satisfy the inequality above is with $\sum_i x_{i,j}  = 1$ for all $j \in [n]$. However, the maximal value of the left-hand side coefficient of this constraint is too large, yielding a running time of $\Omega(n^n)$ in the ILP solvers of Theorem~\ref{thm:ILP-bounded} and Theorem~\ref{thm:ILP-unbounded}. Instead, we utilize the fact that the coefficients in the objective function do not affect the running time of these solvers, and replace the above constraint with a corresponding objective function. The program $\text{ILP}_1$ is then given by:
\begin{alignat}{3}
\tag{$\text{ILP}_1$}
& \text{max} &\quad & \sum_{j\in[n]} n^{j} \sum_{i\in[m]}  x_{i,j} \;-\; \sum_{j\in[n]}  n^{j} \notag \\
&\text{s.t.} &\quad & \sum_{i \in [m]} \sum_{j \in [n]} x_{i,j}  \;=\; n & \qquad \qquad &  \notag  \\
& &\quad & \sum_{j \in [n]} p_jx_{i,j} +s_i \;=\; q_iC & \qquad \qquad & \forall \, i \in [m] \notag  \\
& \qquad \qquad & & x_{i,j}, s_i \in  \mathbb{N} & \qquad \qquad & \forall \, i \in [m],\, j \in[n] \notag
\end{alignat}
We use $Z=Z(x)$ to denote the value of the objective function on some given solution~$x$ for~$\text{ILP}_1$. 

\begin{lemma}
\label{lem:ILP1}
If there exists a feasible solution with $Z \ge 0$ for~$\text{ILP}_1$ then there exists a feasible solution for~$\text{ILP}_0$.
\end{lemma}

\begin{proof}
Suppose there exists a feasible solution $x$ for~$\text{ILP}_1$ with $Z(x) \ge 0$ which is not feasible for~$\text{ILP}_0$. This means that $\sum_{i\in [m]} x_{i,j} \neq 1$ for some $j\in [n]$. Then, as $x$ is feasible for~$\text{ILP}_1$, we have $\sum_{i \in [m]} \sum_{j \in [n]} x_{i,j}=n$, which implies that there are two indices $j_1,j_2 \in [n]$ with $\sum_{i\in [m]} x_{i,j_1} =0$ and $\sum_{i \in [m]} x_{i,j_2} > 1$. Choose $j_1$ and $j_2$ to be the largest such indices. Then $\sum_{i \in [m]} x_{i,j} =1$ for all $j > j_1,j_2$. 

First, suppose that $j_1 > j_2$. Then, again since $\sum_i \sum_jx_{i,j}=n$, we have
$$
\sum^n_{j=1} n^j \;>\; \sum^n_{j=j_1+1} n^j + j_1 \cdot n^{j_1-1} \;>\; \sum^n_{j=1}n^j\sum^m_{i=1}  x_{i,j},
$$
and so $Z < 0$, a contradiction. Thus, it must be that $j_1<j_2$. Consequently, we have $p_{j_1} \leq p_{j_2}$ since the jobs are ordered according to the SPT rule. It follows that we can decrease the value of one the variables $x_{i,j_2}$ with $x_{i,j_2} \geq 1$ by one, and increase the value of $x_{i,j_1}$ (with the same index~$i$) by one. Since $p_{j_1} \leq p_{j_2}$, we still have $\sum_j p_jx_{i,j} \leq q_i C$, and so we can set $s_i$ accordingly and obtain an alternative feasible solution for~$\text{ILP}_1$. Continuing in this way repeatedly, we obtain a feasible solution $x'$ for~$\text{ILP}_1$ with $\sum_i x'_{i,j} =1$ for all $j\in [n]$. Such a solution is also feasible for~$\text{ILP}_0$.  
\end{proof}

According to Lemma~\ref{lem:ILP1}, to determine whether~$\text{ILP}_1$ has a feasible solution (which directly corresponds to a schedule of makespan of at most $C$), it suffices to determine whether~$\text{ILP}_0$ has a feasible solution with non-negative objective value. Let $A$ denote the constraint matrix of~$\text{ILP}_1$. Then $A$ has $m+1$ rows, and the maximum absolute value $a_{\max}$ of any coefficient in~$A$ is~$p_{\max}$. Thus, using the short ILP-solver of Theorem~\ref{thm:ILP-unbounded}, we obtain an algorithm with total running time of $\widetilde{O}(p^{2(m+1)}_{\max}+n)$, as promised by Theorem~\ref{thm:main}.

\section{Two natural variants}
\label{sec:TwoVariants}%

We next show how to extend $\text{ILP}_1$, the program formulated in Section~\ref{sec:Makespan}, so that it captures the \threefield{Qm}{cap}{C_{\max}} and \threefield{Qm}{rej}{C_{\max}} problems as well. Recall that in the former variant, each machine $M_i$ is associated with a bound $b_i$, and we can only schedule at most $b_i$ jobs on this machine. In the latter, each job $j \in J$ has a weight $w_j$, and we can reject (\emph{i.e.} not include in our schedule) jobs with total weight at most some given $W \in \mathbb{N}$. 

To solve \threefield{Qm}{cap}{C_{\max}}, we add to $\text{ILP}_1$ $m$ additional slack variables $\widehat{s}_1,\ldots,\widehat{s}_m$, and~$m$ additional constraints, that together ensure that all capacity constraints are maintained. The constraint corresponding machine~$i$ is written as 
$$
\sum_{j\in [n]}  x_{i,j} \,+\, \widehat{s}_{i} \;=\; b_i.
$$
We call the corresponding modified program $\text{ILP}_2$. It is not difficult to see, again using Lemma~\ref{lem:ILP1}, that feasible solutions of $\text{ILP}_2$ correspond to schedules $\sigma$ of makespan at most $C$ where each machine capacity is maintained, and vice-versa. Since the constraint matrix~$A$ of $\text{ILP}_2$ has $2m+1$ rows, and $a_{\max}=p_{\max}$, by Theorem~\ref{thm:ILP-unbounded} we get: 

\begin{theorem}
\threefield{Qm}{cap}{C_{\max}} can be solved in $\widetilde{O}(p^{4m+2}_{\max} + n)$ time. 
\end{theorem}

Next consider the \threefield{Qm}{rej}{C_{\max}}. For this variant, we assume that we have an additional machine $M_0$ where all rejected jobs are scheduled. Thus, our goal is to find a schedule where $M_i(\sigma) \leq C$ for all $i \in [m]$, and $M_0(\sigma) \leq W$. We therefore add the following constraint, with an additional slack variable~$s_0$, to $\text{ILP}_1$ of the previous section:
$$
\sum_{j\in[n]} w_jx_{0,j} \,+\, s_0 \;=\; W.
$$
Moreover, we modify the objective function and constraints so that index~$i$ now ranges from~$0$ to $m$. Unfortunately, we cannot directly use the \textsc{ILP-unbounded} solver of Theorem~\ref{thm:ILP-unbounded} on our new ILP. This is because there might be feasible solutions with $x_{0,j} > 1$ for some $j \in [n]$ that do not correspond to feasible schedules. Indeed, for a pair of jobs $j_1$ and $j_2$ with $p_{j_1} < p_{j_2}$ and $w_{j_1} > w_{j_2}$, it might be beneficial to reject $j_2$ twice and not schedule $j_1$ at all. For this reason, we introduce upper-bound constraints on the non-slack variables of our ILP, and require that $x_{i,j} \leq 1$ for all $i \in \{0,\ldots,m\}$ and $j\in[n]$. We call this program $\text{ILP}_3$.   

\begin{lemma}
If there exists a feasible solution with $Z \ge 0$ for $\text{ILP}_3$, then there exists a feasible solution $x$ for $\text{ILP}_0$ with $\sum_j w_j x_{0,j} \leq W$.
\end{lemma}

\begin{proof}
Suppose there exists a feasible solution $x$ for $\text{ILP}_3$ with $Z(x) \ge 0$ that is not feasible for $\text{ILP}_0$. Then as $x$ is feasible for $\text{ILP}_3$, we have $\sum_j w_j x_{0,j} \leq W$, and as it is infeasible for $\text{ILP}_0$, we have $\sum_i x_{i,j} \neq 1$ for some $j\in [n]$. Then as $\sum_i \sum_j x_{i,j}=n$ by the first constraint in $\text{ILP}_3$, there must be two indices $j_1,j_2 \in [n]$ with $\sum_i x_{i,j_1} =0$ and $\sum_i x_{i,j_2} > 1$. As is shown in the proof of Lemma~\ref{lem:ILP1}, it must be that $j_1 \leq j_2$ or otherwise $Z(x) <0$. Thus, $p_{j_1} \leq p_{j_2}$ since $J$ is ordered according to the SPT rule.

Now, using the fact that $x_{i,j_2}\leq 1$ for all $i \in \{0,\ldots,m\}$ in $\text{ILP}_3$ due to the upper-bound constraints, it follows that there are two distinct indices $i_1 \neq i_2 \in \{0,\ldots,m\}$ such that $x_{i_1,j_2}=1$ and $x_{i_2,j_2}= 1$. At least one of these indices is not equal to zero, so assume that $i_1 \neq 0$. Since $p_{j_1} \leq p_{j_2}$, we can decrease the value of  $x_{i_1,j_2}$ to zero, and increase the value of $x_{i_1,j_1}$ to one. By adjusting $s_{i_1}$ accordingly, the constraint corresponding to $i_1$ is satisfied, and moreover, $\sum_j w_j x_{0,j}$ remains the same. Thus, we obtain an alternative feasible solution for~$\text{ILP}_3$, and so repeatedly continuing in this fashion we construct a feasible solution $x'$ with $\sum_i x'_{i,j} =1$ for all $j\in [n]$. Such a solution is feasible for $\text{ILP}_0$ and has $\sum_j w_j x'_{0,j} \leq W$.
\end{proof}

Note that the constraint matrix $A$ of $\text{ILP}_3$ has $m+2$ rows, and that 
$$
a_{\max}=\max\{p_{\max},\,w_{\max}\}=O(p_{\max}+w_{\max}).
$$
Thus, by Theorem~\ref{thm:ILP-bounded} we get: 
\begin{theorem}
\threefield{Qm}{rej}{C_{\max}} can be solved in $\widetilde{O}((p_{\max} + w_{\max})^{(m+2)(m+3)} \cdot n)$ time.
\end{theorem}

\section{Release dates}
\label{sec:ReleaseDates}%

Consider now the \threefield{Qm}{r_j\geq 0}{C_{\max}} problem, the extension of \threefield{Qm}{}{C_{\max}} to the case where jobs have arbitrary release dates. Let $r_j$ denote the release date of job $j$, meaning that job~$j$ cannot be processed before time $r_j$ in any feasible schedule. It is well known that there exists an optimal schedule in which jobs are processed in a non-decreasing order of release dates on each machine (as shown by~\cite{Jackson1956} for the single machine counterpart), also known as an earliest release date (ERD) order. Moreover, by processing jobs as late as possible without increasing the makespan, we can assume that all jobs are processed consecutively without gaps starting from some time $t_i \geq 0$ on any machine $M_i$. Thus, given a bound $C$ on the makespan of the schedule, we search for a schedule $\sigma$ such that for each $i \in [m]$, all jobs on machine $i$ start after their release date when scheduled according to the ERD order, where the first job is processed starting at time $C-M_i(\sigma)$.

Let $r_{\#}=\{r_1,\ldots,r_n\}$ be the number of distinct release dates in~$J$, and let $r^{(1)},\ldots,r^{(r_{\#})}$ denote these release dates. We say that two jobs $j_1,j_2\in J$ are of the same \emph{type} if they have the same release date, \emph{i.e.} if~$r_{j_1}=r_{j_2}$. We partition the input job set into types. That is, for $k \in [r_{\#}]$, let $J^{(k)}=\{j \in J : r_j = r^{(k)}\}$ denote the jobs in $J$ with release date $r^{(k)}$, and let $n_k=|J^{(k)}|$ denote the number of these jobs. Here as well, we include one variable $x_{i,j}\in \mathbb{N}$, for any $i\in[m]$ and $j\in[n]$ with the intended interpretation that $x_{i,j}=1$ if job~$j$ is assigned to $M_i$ and $x_{i,j}=0$ otherwise. Here, as well, our ILP has a variable $x_{i,j}$ for each job-machine pair, and a variable $s_{i,k}$ for each type-machine pair. 

Next, we describe the constraints in our formulation. For each $i \in [m]$ and $k \in [r_\#]$, we include the following constraint that ensures the all jobs in $J^{(k)}$ scheduled on machine $M_i$ are not processed before their release time~$r_k$:
$$
\sum^{r_{\#}}_{k'=k} \sum_{j \in J^{(k')}}  p_j  x_{i,j} \,+\, s_{i,k} \;=\; q_iC-r^{(k)}.
$$
Next we split the equality $\sum_i \sum _j x_{i,j} = n$ used in the previous formulations into $r_{\#}$ equalities. For each $k \in [r_{\#}]$, we add a constraint corresponding to $J^{(k)}$ which is written as:
$$
\sum_{i \in [m]} \sum_{j \in J^{(k)}} x_{i,j}  \;=\; n_k.
$$
Finally, we add the same objective functions used in the previous formulations.


In total, we obtain an ILP with $r_{\#}(m+1)$ constraints (not counting the integrality constraints), which we refer to as $\text{ILP}_4$:
\begin{alignat}{3}
\tag{$\text{ILP}_4$}%
& \text{max} &\quad & \sum_{j\in[n]} n^{j} \sum_{i\in[m]}  x_{i,j} \;-\; \sum_{j\in[n]}  n^{j} \notag & & \\
&\text{s.t.} &\quad & \sum_{i \in [m]} \sum_{j \in J^{(k)}} x_{i,j}  \;=\; n_k  & \qquad \qquad & \forall  k \in [r_{\#}]     \notag \\
& &\quad & \sum^{r_{\#}}_{k'=k} \sum_{j \in \mathcal{S}^{(k')}}  p_j  x_{i,j} \,+\, s_{i,k} \;=\; q_iC-r^{(k)} & \qquad \qquad & \forall \, i \in [m],\, k \in [r_{\#}] \notag  \\
& & & x_{i,j}, s_{i,k} \in  \mathbb{N} & \qquad \qquad & \forall \, i \in [m],\, j \in[n],\, k \in[r_{\#}] \notag
\end{alignat}
It is not difficult to see that any feasible solution $x$ for $\text{ILP}_4$ with $\sum_i x_{i,j} =1$ for all $j\in [n]$ corresponds to a schedule for of makespan at most $C$. In the following we show that the existence of such a solution is equivalent to the existence of a solution with non-negative objective value. 

\begin{lemma}
\label{lem:ReleaseDates}%
If there exists a feasible solution with $Z \ge 0$ to $\text{ILP}_4$ then there exists a feasible solution $x$ with $\sum_i x_{i,j} =1$ for all $j\in [n]$.
\end{lemma}

\begin{proof}
Suppose there exists a feasible solution $x$ to the above $\text{ILP}_4$ with $Z(x) \ge 0$ and with $\sum_i x_{i,j} \neq 1$ for some $j\in J$. As $x$ is feasible for $\text{ILP}_4$, we have $\sum_{j \in J^{(k)}} \sum_{i \in [m]} x_{i,j} = n_k$. Thus, it must be that there are two job indices $j_1,j_2  \in J^{(k)}$ with $\sum_{i \in [m]} x_{i,j_1}  =0$ and $\sum_{i \in [m]} x_{i,j_2} > 1$. As shown in the proof of Lemma~\ref{lem:ILP1}, we have $j_1<j_2$ or otherwise $Z(x) < 0$. Consequently, we have $p_{j_1} \leq p_{j_2}$ due to the SPT ordering of the jobs, and so we can decrease the value of some $x_{i,j_2}$ with $x_{i,j_2} \geq 1$ by one, and increase the value of $x_{i,j_1}$ by one, without violating the constraint corresponding to $M_i$ (following an appropriate adjustment of~$s_i$). Continuing in this fashion, we obtain an alternative feasible solution~$x'$ to $\text{ILP}_4$ with $\sum_{i \in[n]} x'_{i,j} =1$ for all $j\in [n]$.  
\end{proof}

Since the constraint matrix of $\text{ILP}_4$ has $mr_{\#}+r_{\#}$ rows, and $A_{\max}=p_{\max}$, using Theorem~\ref{thm:ILP-unbounded} we obtain: 

\begin{theorem}
\threefield{Qm}{r_j}{C_{\max}} can be solved in $\widetilde{O}(p^{2r_{\#}(m+1)}_{\max} + r_{\#}\cdot n)$ time. 
\end{theorem}

\section{Unrelated Machines}

We next to consider the general case of unrelated machines. In the following we show that \threefield{Rm}{p_{i,j} \in \{p_j,\infty\}}{C_{\max}} problem admits an $\widetilde{O}(p^{2^{m+1}+2m-2}_{\max} + n)$ time algorithm, while \threefield{R2}{}{C_{\max}} admits an algorithm running in time $\widetilde{O}(p^{6}_{\max} \cdot n)$.

We begin with the \threefield{Rm}{p_{i,j} \in \{p_j,\infty\}}{C_{\max}} problem. Recall that in this variant, the processing time of job $j$ on each machine has one of two values $ \{p_j,\infty\}$. In other words, job~$j$ cannot be scheduled on machines $M_i$ with $p_{i,j}=\infty$, and machine~$M_i$ cannot process jobs~$j$ with $p_{i,j}=\infty$. For $i \in [m]$, let $J(M_i)=\{j \in J: p_{i,j} \neq \infty \}$ denote the set of jobs that can processed on machine $M_i$, and for $j \in [n]$, let $M(j)=\{i \in [m] : p_{i,j} \neq \infty\}$ denote the set of machines that can process job~$j$. As in previous sections, we assume that $p_1 \leq p_2,\ldots,\leq p_n$.  

Our short ILP will include a variable $x_{i,j}$ for each job $j \in [n]$, and machine $i \in M(j)$. Moreover, as usual, we have $m$ additional slack variables $s_1,\ldots,s_m$. We add a makespan constraint for each $M_i$ which is written as:
$$
\sum_{j \in J(M_i)} p_j  x_{i,j} \,+\, s_{i} \;=\; C.
$$
Next, as is done in Section~\ref{sec:ReleaseDates}, we partition the jobs into types. We say that two jobs $j_1,j_2 \in J$ are of the same \emph{type} if they can be processed on the same set of machines, \emph{i.e} if $M(j_1)=M(j_2)$. Let $t_{\#}$ denote the number of different job types in $J$, and let $J^{(1)},\ldots,J^{(t_{\#})}$ denote the equivalence classes in our partitioning of $J$ into types. For each $k \in [t_{\#}]$, we add to our ILP a constraint corresponding to~$J^{(k)}$ which is written as:
$$
\sum_{j \in J^{(k)}} \sum_{i \in M(j)} x_{i,j}  \;=\; n_k,
$$
where $n_k=|J^{(k)}|$. 

Our short ILP formulation for \threefield{Rm}{p_{i,j} \in \{p_j,\infty\}}{C_{\max}}, which we denote by $\text{ILP}_5$, is written as follows:   
\begin{alignat}{3}
\tag{$\text{ILP}_5$}%
& \text{max} &\quad & \sum_{j\in[n]} n^{j} \sum_{i\in[m]}  x_{i,j} \;-\; \sum_{j\in[n]}  n^{j} \notag & & \\
&\text{s.t.} &\quad & \sum_{j \in J^{(k)}} \sum_{i \in M(j)} x_{i,j}  \;=\; n_k  & \qquad \qquad & \forall  k \in [t_{\#}]     \notag \\
& &\quad & \sum_{j \in J(M_i)} p_j  x_{i,j} \,+\, s_{i} \;=\; C & \qquad \qquad & \forall \, i \in [m] \notag  \\
& & & x_{i,j}\in  \mathbb{N} & \qquad \qquad & \forall \,j \in[n] \text{ and } i \in M(j) \notag\\
& & & s_i\in  \mathbb{N} & \qquad \qquad & \forall \,i \in[m] \notag
\end{alignat}
It is easy to see that any feasible solution for $\text{ILP}_5$ with $\sum_{i \in M(j)} x_{i,j}=1$ for all $j\in [n]$ corresponds to a schedule for $J$ with makespan at most $C$. The following lemma, whose proof is similar to the proof of Lemma~\ref{lem:ReleaseDates}, shows that any feasible solution with non-negative objective value can be converted into such a feasible solution. 

\begin{lemma}
\label{lem:Eligable}
If there exists a feasible solution with $Z \ge 0$ to $\text{ILP}_5$ then there exists a feasible solution $x$ with $\sum_{i \in M(j)} x_{i,j}=1$ for all $j\in [n]$.
\end{lemma}

\begin{proof}
Suppose there exists a feasible solution $x$ to the above $\text{ILP}_5$ with $Z(x) \ge 0$ and with $\sum_{i \in M(j)} x_{i,j}\neq 1$ for some $j\in [n]$. Consider the largest index~$j$ with $\sum_{i \in M(j)} x_{i,j} \neq 1$ holds, and assume let $k$ be such that $j \in J^{(k)}$. As $x$ is feasible, we have $\sum_{j \in J^{(k)}} \sum_{i \in M(j)} x_{i,j} = n_k$. Thus, it must be that there are two job indices $j_1,j_2  \in J^{(k)}$ with $\sum_{i \in M(j)} x_{i,j_1}  =0$, $\sum_{i \in M(j)} x_{i,j_2} > 1$. As shown in the proof of Lemma~\ref{lem:ILP1}, we have $j_1<j_2$ since otherwise $Z(x) < 0$. Consequently, we get $p_{j_1} \leq p_{j_2}$ due to the SPT ordering of the jobs. Moreover, since $j_1$ and $j_2$ are of the same type, we have $M(j_1)=M(j_2)$. Thus, we can decrease the value of some $x_{i,j_2}$ with $i \in M(j_2)$ and $x_{i,j_2} \geq 1$ by one, and increase the value of $x_{i,j_1}$ by one, as $i \in M(j_1)$ as well. This does not increase $\sum_{j \in J(M_i)} p_j  x_{i,j}$, and so we adjust~$s_i$ to obtain another feasible solution. Continuing in this fashion, we obtain an alternative feasible solution~$x'$ to $\text{ILP}_5$ with $\sum_{i \in M(j)} x'_{i,j}=1$ for all~$j\in [n]$.   
\end{proof}

Note that the constraint matrix $A$ of $\text{ILP}_5$ has $m+t_{\#}$ rows, with $a_{\max}=p_{\max}$. As $M(j) \subseteq [m]$ for each job $j \in [n]$, and we may assume that $M(j) \neq \emptyset$ for any $j \in [n]$, we have that 
$$
t_{\#} = |\{M(j) : j \in J\}|\leq 2^m - 1.
$$
Plugging this into Theorem~\ref{thm:ILP-unbounded}, we obtain the following result.

\begin{theorem}
\threefield{Rm}{p_{i,j} \in \{p_j,\infty\}}{C_{\max}}  can be solved in $\widetilde{O}(p^{2^{m+1}+2m-2)}_{\max} + n)$ time. 
\end{theorem}

We next consider \threefield{R2}{}{C_{\max}}. The fact that we were able to use a compact formalization in the previous sections, which enabled us to apply short ILP solvers, relied heavily on the fact that the SPT order is independent of the identity of the machine. This is not the case, unfortunately, when machines are unrelated. However, when there are only two machines present, a different approach still gives us a short ILP formulation. 

Here, we define a single variable $x_j$ associated with each job~$j$, where our intended interpretation is that $x_j=1$ implies a schedule~$\sigma$ with $\sigma(j)=1$, and  $x_j=0$ implies $\sigma(j)=2$. Since $\sum p_{1,j} x_j$ equals the makespan on $M_1$, and $\sum p_{2,j} (1-x_j)$ equals the makespan on $M_2$, there is a schedule with makespan at most $C$ if and only if there is a feasible solution for the following program:
\begin{alignat}{3}
\tag{$\text{ILP}_6$}%
& \quad &\quad &  \sum_{j \in [n]}  p_{1,j}  x_{j} \,+\, s_{1} \;=\;C  \notag & & \\
& & & \sum_{j \in [n]}  p_{2,j}  (1-x_{j}) \,+\, s_{2} \;=\;C  & \qquad \qquad & \notag \\
& & & x_{j}\in  \{0,1\} & \qquad \qquad & \forall \,j \in[n] \notag\\
& & & s_1,s_2\in  \mathbb{N} & \qquad \qquad & \notag
\end{alignat}
Note that here we cannot relax the binary constraints, as a feasible solution might be constructed by setting $x_j=2$ in cases where $p_{2,j}$ is sufficiently smallr than~$p_{1,j}$. Accordingly, we apply \textsc{ILP-bounded} to solve the above formulation. Since $\text{ILP}_6$ has 2 constraints, the following holds due to Theorem~\ref{thm:ILP-bounded}: 
\begin{corollary}
\threefield{R2}{}{C_{\max}} can be solved in $\widetilde{O}(p^{6}_{\max} \cdot n)$ time. 
\end{corollary}

\section{Discussion}

In this paper we showed how to use short ILPs to obtain fast pseudo-polynomial time algorithms for makespan minimization on a fixed number of machines, and other generalizations of this problem. We believe that short ILPs can be of use for other scheduling problems. For instance, the \threefield{Qm}{r_j \geq 0}{C_{\max}} is equivalent to the \threefield{Qm}{}{L_{\max}} problem~\cite{LenstraShmoys2020}, the problem of minimizing the maximal lateness on $m$ related machines, so the algorithm described in Section~\ref{sec:TwoVariants} applies to this problem as well. The \threefield{Qm}{rej}{C_{\max}} problem is equivalent to \threefield{Qm}{d_j = d}{\sum w_j U_j}, the problem of minimizing the weighted number of tardy jobs on $m$ related machines with common due date for all jobs, yielding an $\widetilde{O}((p_{\max}+w_{\max})^{(m+2)(m+3)} \cdot n)$ time algorithm for this problem as well. It is interesting to find other scheduling problems where short ILPs can prove useful. Finally, the reader might have observed that most of our ILP formulations (specifically, those that used variables for each job-machine pair) are quite sparse in the sense that of the most coefficients of the constraint matrix are zero. A natural question to ask is whether this sparsity can be leveraged to obtain algorithms that improve on Theorem~\ref{thm:ILP-bounded} and Theorem~\ref{thm:ILP-unbounded}.


\bibliographystyle{abbrvnat}
\bibliography{biblio}

\end{document}